\documentclass[pra,aps,amsmath,amssymb,superscriptaddress,twocolumn,longbibliography]{revtex4-1}
\usepackage{graphicx,multirow}
\usepackage{color,outlines}
\usepackage[english]{babel}
\usepackage{amsthm}
\usepackage[sc,osf]{mathpazo}
\usepackage{hyperref}
\usepackage{physics}
\usepackage{outlines}
\usepackage{parskip}

\newcommand{\bla}{\color{black}}


\newcounter{parentnumber}

\newtheorem{thm}{Theorem}

\definecolor{green2}{RGB}{0,100,0}
\newcommand{\md}{\mathcal{D}} 
\newcommand{\e}{\mathcal{E}}

\begin{document}

\title{Generalized $\alpha$-Observational Entropy}

\author{Shivam Sinha}
\email{shivamquantum@gmail.com}
\author{Nripendra Majumdar}
\email{nripendrajoin123@gmail.com}
\author{S. Aravinda }
\email{aravinda@iittp.ac.in}
\affiliation{Indian Institute of Technology Tirupati, Tirupati, India~517619}  
\begin{abstract}

   Traditional measures of entropy, like the von Neumann entropy, while fundamental in quantum information theory, are insufficient when interpreted as thermodynamic entropy due to their invariance under unitary transformations, which contradicts observed entropy increases in isolated systems. 
   Recognizing this limitations of existing measures for thermodynamic entropy, recent research has focused on observational entropy (OE) as a promising alternative, offering practical applicability and theoretical insights. In this work, we extend the scope of observational entropy by generalizing it to a parameterized version called  $\alpha$-Observational entropy ($\alpha$-OE). $\alpha$-OE is expressed in terms of the Petz-R\'{e}nyi relative entropy between the states on which a quantum-to-classical channel is applied. The $\alpha$-OE reduces to OE under $\alpha\rightarrow 1$. We prove various properties of the $\alpha$-OE, which are the generalization of the properties of OE, including the monotonically increasing nature of $\alpha$-OE as a function of refinement of coarse-graining. 
   We further explore the role of $\alpha$-OE in thermodynamic contexts, particularly for the entropy production in open and closed quantum systems and its relation with the Helmholtz free energy. 
\end{abstract}

\maketitle

\section{Introduction} 

In the landscape of quantum information theory, the von Neumann entropy has long been a cornerstone, a powerful tool for quantifying the uncertainty and complexity of the quantum system. However, it faces a fundamental limitation when interpreted as thermodynamic entropy~\cite{von2010proof, She99, HS06, MMB05, DRE+11, AY13}. The contradiction is due to the invariance of von Neumann entropy under unitary transformation. This contradicts the empirical reality of entropy increases observed in isolated systems, such as free gas expansion or irreversible mixing of substances. von Neumann admitted the discrepancy and knew that this entropy generally cannot be directly equated with thermodynamic entropy. He introduced a novel entropy measure, which is apt for thermodynamic scenarios, but its importance has been forgotten and totally overshadowed by his more famous creation— von Neumann's entropy; he named it "Macroscopic entropy"~\cite{von2018mathematical}. Safranek, Deutsch, and Aguirre recently reintroduced this entropic quantity as observational entropy. It generalizes the Gibbs and Boltzmann entropy~\cite{deutchquantum2019, quantum2020deautch}, and can be interpolated between the two entropies. The form of OE has been present in the literature in various works with different names~\cite{gibbs1902elementary,ehrenfest1990conceptual,von2010proof, Caves_notes}. The OE is defined with respect to a measurement called coarse-graining, and it is shown that the OE is a monotonic function of the refinement of the coarse-graining, reflecting a reason for the emergence of irreversibility or an increase in entropy from the perspective of the Second law of thermodynamics~~\cite{vsafranek2020classical,deutchquantum2019,deutchclassical2020,quantum2020deautch}. 
The OE has been applied in various scenarios, particularly in understanding the laws of thermodynamics, quantum chaos, and in many-body systems~\cite{vsafranek2021brief,strasberg2021second,vsafranek2020classical, Buscemi_22,schindler2020quantum,strasberg2021second,sreeram2023witnessing,modak2022observational,pg2023periodicity}.

In a recent work~\cite{buscemi2022observational}, the OE has been defined using a quantum-to-classical channel and the Umegaki quantum relative entropy. Many properties of OE have been shown using the data-processing inequality of quantum relative entropy and the equality due to the Petz recovery map. Our main goal of this work is to define the $\alpha$-generalization  of OE as $\alpha$-OE and study its properties. One of the most famous generalizations of Umegaki quantum relative entropy is the Petz-R\'{e}nyi generalization called $\alpha$ quantum relative entropy~\cite{petz1986quasi}, which has found various applications in quantum information theory~\cite{zhu2017coherence,polyanskiy2010arimoto,misra2015quantum,bao2019holographic,zhu2022generalized,ogawa2000strong,ogawa2004error,hayashi2007error,audenaert2008asymptotic,koenig2009strong,mosonyi2011quantum,mosonyi2015quantum,
cooney2016strong,hayashi2016correlation, tomamichel2016strong, ding2018strong, leditzky2016strong, wilde2017converse}. 

 We define $\alpha$-OE using Petz-Rényi relative entropy and a measurement channel, often referred to in the literature as a 'quantum-to-classical channel.' This term reflects the process of associating quantum states with classical probabilistic outcomes derived from measurements. While the resultant state retains quantum properties, the term emphasizes the extraction of classical information about the measurement outcomes. 
 The $\alpha$-OE reduces to OE under $\alpha\rightarrow 1$. We generalize all the properties of OE to $\alpha$-OE by using  Petz-R\'{e}nyi quantum relative entropy and further study the implications of the $\alpha$-OE in the context of thermodynamic entropy. We have used the $\alpha$-OE for calculating the entropy production in closed and open quantum systems, thereby formulating second law-like formulations~\cite{strasberg2021second}. The entropy production and the $\alpha$-OE  is then related to various notions of thermodynamic quantities.

The paper is organized as follows. In Sec.~(\ref{sec:oe}) we define $\alpha$-OE and relate it to quantum relative entropies. We prove various properties of $\alpha$-OE in Sec.~(\ref{sec:prop}). The refinement of the coarse-graining is defined in Sec.~(\ref{sec:refine}), and the $\alpha$-OE is shown to be a monotonic function of refinement. The concept of Sequential Coarse-Graining and its impact on $\alpha$-OE is in Sec.~(\ref{sec:sequenctial}), where we prove that $\alpha$-OE always decreases under sequential measurements, similar to the standard observational entropy. The coarse-grained state is defined in Sec.~(\ref{sec:cg}), and the condition under which the  $\alpha$-OE is equal to R\'{e}nyi entropy is studied. Lastly, Sec.~(\ref{sec:entropyproduction}) delves into the implications of $\alpha$-OE for entropy production in both closed and open quantum systems, including an analysis of quantum analogs of the Second Law of Thermodynamics and Clausius' law in terms of $\alpha$-OE. We also show how $\alpha$-OE is related to the Helmholtz free energy, demonstrating its broader applicability in thermodynamic contexts. The mathematical proofs are presented in the Appendix. 

\bla 

\section{$\alpha$-Observational entropy  \label{sec:oe}}
Consider a quantum system $\rho$ represented as a positive trace one operator defined on the finite-dimensional Hilbert space $\mathcal{H}_d$ of dimension $d$.  The set of positive-semidefinite operators $\Pi_i\geq 0$, with $\sum_i \Pi_i = I$, called a positive operator valued measure (POVM), forms a measurement.  The  set of POVM's $\{\Pi_i\}_i$ with $\sum_i \Pi_i = I$ is called a {\it coarse-graining}, denoted as $\chi$.  The idea of coarse-graining is that the identity can be decomposed by the set of POVM in various ways, and each such decomposition forms a coarse-graining $\chi$.  We will explain later the relation between various coarse-grainings.

The observational entropy (OE) for the state $\rho$, with respect to a  coarse-grainin $\chi$ defined by  associated POVM $\{\Pi_i\}_i$, is given by
\begin{equation}
    S_\chi (\rho) = -\sum_{i} p_i \log \frac{p_i}{V_i},
\end{equation}
where $p_i = \Tr (\Pi_i \rho)$ is the probability of finding the system in the subspace corresponding to the  $\Pi_i$, and $V_i = \Tr (\Pi_i)$ is the volume of the subspace corresponding to the set of $\Pi_i$.

Define a quantum channel $\e$, called a quantum-to-classical channel, which projects the system into the orthogonal subspaces of the coarse-graining $\chi$, 
\begin{equation}
    \e (\bullet) = \sum_i \Tr (\Pi_i \bullet ) \ketbra{i}. 
    \label{eq:mchan}
\end{equation}
The OE can be expressed in terms of the {\it Umegaki quantum relative entropy}~\cite{umegaki1962conditional} $\md$, which is  defined for any two states $\rho$ and $\sigma$ as
\begin{equation}
\md(\rho || \sigma) = 
\begin{cases}
\Tr (\rho \log \rho) - \Tr(\rho \log \sigma), &  \text{ supp}(\rho) \subseteq \text{supp} (\sigma) \\
\infty & \text{else}. 
\end{cases}    
\end{equation}
The von Neumann entropy $S_{vN}(\rho)$ for any state $\rho$ is 
\begin{equation}
    S_{vN}(\rho) := -\Tr (\rho \log \rho ), 
    \label{eq:vn}
\end{equation}
and from the quantum relative entropy, this can be written as 
\begin{equation}
    S_{vN}(\rho) = -\md(\rho||I) = \log d - \md (\rho||I_d), 
\end{equation}
where $I$ is the identity matrix and $I_d = I/d$, a maximally mixed state. 

Many interesting properties of OE are obtained by expressing it in terms of quantum relative entropy. The action of measurement channel $\e$ on the state $\rho$ and on the identity matrix $I$ is $\e(\rho) = \sum_i p_i \ketbra{i}$ and $\e(I) = \sum_i V_i \ketbra{i}$ respectively. By using this it can be readily seen~\cite{Buscemi_22} 
\begin{equation}
    S_\chi (\rho) = \log d - \md (\e(\rho) ||\e(I_d)), 
    \label{eq:chn}
\end{equation}
and 
\begin{equation}
     S_\chi (\rho) - S_{vN} (\rho) = \md (\rho||I_d) -  \md (\e(\rho)||\e(I_d)). 
\end{equation}
Quantum relative entropy is monotonous under the action of a quantum channel~\cite{lindblad1975completely}.  For a given quantum channel $\mathcal{N}$, and for any two states $\rho$ and $\sigma$, 
\begin{equation}
    \md (\rho||\sigma) \geq  \md (\mathcal{N}(\rho)||\mathcal{N}(\sigma)), 
    \label{eq:data}
\end{equation}
from this it can be seen that $S_\chi (\rho) - S_{vN} (\rho) \geq 0$. 

Here we would like to ask what is the corresponding generalization of OE in terms of the $\alpha$ generalization of quantum relative entropy. There can be many generalizations of quantum relative entropy, which reduce to quantum relative entropy at suitable limits. The Petz-R\'{e}nyi relative entropy $\md_\alpha$~\cite{petz1998contraction,petz1986quasi}  of   any two states $\rho$ and $\sigma$  is given as 
\begin{equation}
\md_\alpha (\rho||\sigma) := 
\begin{cases}
    \frac{1}{\alpha-1} \log \Tr \left(\rho^\alpha \sigma^{1-\alpha}\right),  &  \text{ supp}(\rho) \subseteq \text{supp} (\sigma) \\
\infty & \text{else}. \label{eq:petz}
\end{cases}
\end{equation}
The Petz-R\'{e}nyi relative entropy reproduces quantum relative entropy in the limit $\alpha \rightarrow 1$, and monotonic w.r.t to the quantum channels $\mathcal{N}$,
\begin{equation}
\begin{split}
    &\lim_{\alpha \rightarrow 1}  \md_\alpha (\rho||\sigma) =  \md (\rho||\sigma),\\ 
    &\md_\alpha (\rho||\sigma) \geq  \md_\alpha (\mathcal{N}(\rho)||\mathcal{N}(\sigma)).
\end{split}
\end{equation}

Ordering property of the Petz-R\'{e}nyi relative entropy~\cite{van2014renyi},
\begin{equation}
     \md_{\alpha_1} (\rho||\sigma) \geq \md_{\alpha_2} (\rho||\sigma), \quad \forall \ \ \ \alpha_1 > \alpha_2>0.\label{eq:D_alpha_ordering}
\end{equation}

We define the generalization of observational entropy corresponding to the coarse-graining $\chi$ as
\begin{equation}
    S_\chi^\alpha (\rho) := - \frac{1}{\alpha -1} \log \sum_i p_i^\alpha V_i^{1-\alpha},
    \label{eq:goe}
\end{equation}
 of order $\alpha$, where $\alpha\in(0,1)\cup(1,\infty)$ which we call here as generalized $\alpha$-observational entropy. For simplicity, we use the term $\alpha$-OE throughout this work. The choice of $\alpha$ made on the basis as our $\alpha$-OE is defined using Petz-R\'{e}nyi relative entropy and it recovers the standard observational entropy in the limit $\alpha \to 1$, as discussed in the following section. 
 
 \section{Properties of $\alpha$-OE \label{sec:prop}}

In this Section, we prove various properties of $\alpha$-OE, including its limiting behavior as $\alpha \to 1$, its relationship with Petz R\'{e}nyi relative entropy, ordering, and additivity properties.
\begin{thm}
    The $\alpha$-OE satisfies the following properties.
\end{thm}
 
\begin{outline}[enumerate]
\1  The limiting case of $\alpha$-OE as $\alpha \rightarrow 1$ is given as OE as follows 
\begin{equation}
    \lim_{\alpha \rightarrow 1} S_\chi^\alpha (\rho) = S_\chi (\rho) := - \sum_i p_i \log \frac{p_i}{V_i},
    \end{equation} 
For $\alpha \rightarrow 1$, $\alpha$-OE in Eq.~(\ref{eq:goe}) is of the indeterminate form and by using L'H\^{o}pital's rule by differentiating the function in Eq.~(\ref{eq:goe}), 
\begin{equation}
    \begin{split}
        &=-\dv{\alpha}\log\sum_i p_i^\alpha V_i^{1-\alpha}  \\
        & = -\frac{1}{\sum_i p_i^\alpha V_i^{1-\alpha}} \sum_i p_i^\alpha V_i^{1-\alpha} \log \frac{p_i}{V_i} \\
        & \text{as} \quad \alpha \rightarrow 1 \\
        & = -\sum_i p_i \log \frac{p_i}{V_i}.
    \end{split}
\end{equation}

\1  Similar to expressing OE in the form of relative entropy~\cite{Buscemi_22}, we can  express  $\alpha$-OE in terms of R\'{e}nyi $\alpha$ relative entropy as follows: 
\begin{equation}
    \begin{split}
        S_\chi^\alpha  (\rho) &= -\md_\alpha(\e(\rho)||\e(I))\\ 
        &= \log d - \md_\alpha (\e(\rho) ||\e(I_d)).\label{eq:obser_petz}
    \end{split}   
\end{equation} 

\1  For any quantum state $\rho$, the  R\'{e}nyi entropy of order $\alpha$ is~\cite{renyi1961measures},
\begin{equation}
\begin{split}
    S^\alpha_R(\rho) & =-\frac{1}{\alpha-1}\log{\Tr{\rho^\alpha}} \\
    & = -\md_\alpha(\rho||I) \\
    & = \log d - \md_\alpha(\rho||I_d).
\end{split}
\label{eq:alpha_v_N}
\end{equation}
From this we have the following expression 
\begin{equation}
        S_\chi^\alpha (\rho) - S^\alpha_R (\rho)  = \md_\alpha (\rho||I_d) -  \md_\alpha (\e(\rho)||\e(I_d)) \label{eq:chi-vN},
    \end{equation}
and with the data-processing inequality on $\md_\alpha$, it follows that 
\begin{equation}
     S_\chi^\alpha (\rho) \geq S^\alpha_R (\rho).
     \label{eq:diff}
\end{equation}

   \1 For any two  parameters $\alpha_1$ and $\alpha_2$ such that $\alpha_1 > \alpha_2 >0$, for any quantum state $\rho$ and coarse-graining $\chi=\{\Pi_i\}$,
    \begin{equation}
        S_\chi^{\alpha_1}(\rho) \leq S_\chi^{\alpha_2}(\rho).
    \end{equation}
    The proof follows from the definition~(\ref{eq:obser_petz}) and Eq.~(\ref{eq:D_alpha_ordering}).

\1 Let's turn to $\alpha$-OE in the many-particle scenarios. Consider the system $\rho = \rho_1 \otimes \rho_2 \otimes \cdots \otimes \rho_m$ is defined over the Hilbert space $\mathcal{H} = \mathcal{H}_1 \otimes \mathcal{H}_2 \otimes \cdots \otimes \mathcal{H}_m$. The local coarse-graining is defined as $\bar{\chi} = \chi_1 \otimes \chi_2 \otimes \cdots \otimes \chi_m = \{\Pi_{j_1}\otimes \Pi_{j_2}\otimes \cdots \otimes \Pi_{j_m} \}_{j_1,j_2,\cdots j_m}$, then  
\begin{equation}
        S^\alpha_{\bar{\chi}}(\rho)=\sum_{j=1}^{m}S^\alpha_{\chi_j}(\rho_j),
\end{equation} 
Also, 
%\vspace{12pt}
\begin{widetext}
\begin{equation}
    \begin{split}
    &\sum_{j=1}^{m}S^\alpha_{\chi_j}(\rho_j)=S^\alpha_{\chi_1}(\rho_1)+S^\alpha_{\chi_2}(\rho_2)+ \cdots + S^\alpha_{\chi_m}(\rho_m)\\
        &=-\md_\alpha(\e_1(\rho_1) \otimes \cdots \otimes\e_m(\rho_m)||\e_1(I_1)\otimes \cdots \otimes\e_m(I_m))\\
        & = -\md_\alpha((\e_1 \otimes \cdots \otimes \e_m) (\rho_1 \otimes \cdots \otimes \rho_m)||(\e_1 \otimes \cdots \otimes \e_m) (I_1 \otimes \cdots \otimes  I_m)) \\
        &=-\md_\alpha(\e_{\bar{\chi}}(\rho)||\e_{\bar{\chi}}(I)).
    \end{split}
\end{equation}
\end{widetext}
\bla

\1 Let $\mathcal{H}_a$ and $\mathcal{H}_b$ be two Hilbert spaces. Suppose there are density operators $\rho$ on $\mathcal{H}_a$ with coarse-graining $\chi_1$ and $\sigma$ on $\mathcal{H}_b$ with coarse-graining $\chi_2$, then, 
\begin{equation}
    S_{\chi_1 \otimes \chi_2}^\alpha(\rho\otimes\sigma)=S^\alpha_{\chi_1}(\rho) + S^\alpha_{\chi_2}(\sigma). 
\end{equation}
% \begin{proof}
%     \begin{equation}
%         S_{\chi_1 \otimes \chi_2}^\alpha(\rho\otimes\sigma)& = - \md_\alpha(\e_1(\rho)||\e_1(I_1)) - \md_\alpha(\e_2(\rho)||\e_2(I_2))\\
%         &= - \md_\alpha(\e_1(\rho)\otimes \e_2(\sigma)|| \e_1(I_1)\otimes\e_2(I_2))
%     \end{equation}
% \end{proof}
\begin{proof}
    \begin{equation}
        \begin{split}
            S_{\chi_1 \otimes \chi_2}^\alpha(\rho\otimes\sigma)& = - \md_\alpha(\e_1(\rho)||\e_1(I_1)) - \md_\alpha(\e_2(\rho)||\e_2(I_2))\\
            &= - \md_\alpha(\e_1(\rho)\otimes \e_2(\sigma)|| \e_1(I_1)\otimes\e_2(I_2))\\
            &= - \md_\alpha((\e_1\otimes \e_2)(\rho \otimes \sigma)|| (\e_1\otimes\e_2)(I_1 \otimes I_2))\\
            &= S_{\chi_1 \otimes \chi_2}^\alpha(\rho\otimes\sigma).
        \end{split}
    \end{equation}
\end{proof}

\1 Concavity of $\alpha$-OE: For any state $\rho=\sum_{i}\lambda_i\rho_i$, where $\lambda_i$ is the probability distribution and $\rho_i$ is the set of density operators, then
    \begin{equation}
        \sum_{i}\lambda_i S^\alpha_\chi(\rho_i) \le S^\alpha_\chi(\sum_{i}\lambda_i\rho_i).
    \end{equation}

\begin{proof}
    The concavity property of $\alpha$-OE can be derived using the joint convexity of R\'{e}nyi divergence, which has already been proven in \cite{van2014renyi,mosonyi2014convexity}. For $\alpha\in(0,1)$, the R\'{e}nyi divergence is jointly convex and for $\alpha\in(1,\infty)$, the R\'{e}nyi divergence is jointly quasi-convex.\\
    The $\alpha$-OE can be expressed as the negative of R\'{e}nyi Divergence; therefore, we infer that the $\alpha$-OE is jointly concave for $\alpha\in(0,1)$ and quasi-concave for $\alpha>1$. 
\end{proof}

\end{outline}

The $S^{\alpha}_{\chi}(\rho)$ is non-increasing in $\alpha$ for any given distribution, and this can be proven by differentiation as follows: 

\begin{thm}\label{thm:nonpos}
The $\alpha$-OE is a non-increasing function of the parameter $\alpha$ for any coarse-graining $\chi$ and state $\rho$. This monotonicity property follows from the fact that the derivative of $S^\alpha_\chi(\rho)$ with respect to $\alpha$ is non-positive:
\begin{equation}
    \dv{S^\alpha_\chi(\rho)}{\alpha} = -\frac{1}{(\alpha-1)^2} D(x || p),
\end{equation}
where $D(x || p)$ is the Kullback-Leibler divergence between the distributions $x$ and $p$, defined as:
\begin{equation}
    D(x || p) = \sum_i x_i \log \frac{x_i}{p_i},
\end{equation}
and $x_i = \frac{t_i^\alpha V_i}{\sum_i t_i^\alpha V_i}$, with $t_i = p_i / V_i$. Since $D(x || p) \geq 0$, the derivative is non-positive, ensuring $S^\alpha_\chi(\rho)$ is non-increasing.
\end{thm}

\begin{proof}
    The proof is provided in the Appendix.~(\ref{app:nonpos}). 
\end{proof}

\bla 
\subsection*{Sequential Coarse graining \label{sec:sequenctial}}
Suppose a coarse-graining $\chi_1 = \{\Pi_i^{(1)}\}$ is followed by $\chi_2 = \{\Pi_j^{(2)}\}$. The sequential coarse-graining $\chi_2\chi_1 = \{\Pi_{ij}\}$ and the corresponding probabilities and the volume term are given as $p_{ij} = \Tr (\Pi_{ij}\rho)$ and $V_{ij} = \Tr (\Pi_{ij})$.

\begin{thm} \label{thm:seq}
    For coarse-graining $\chi_1$, $\chi_2$ with the sequential coarse-graining $\chi_2\chi_1$, 
    \begin{equation}
        S_{\chi_1}^\alpha \geq S_{\chi_2\chi_1}^\alpha , 
    \end{equation}
    and the equality holds if and only if 
    \begin{equation}
        \frac{p_i}{V_i}=\frac{p_{ij}}{V_{ij}}. 
    \end{equation}\label{th:sequential}
\end{thm}

\begin{proof}
    The proof is presented in the Appendix.~(\ref{app:seq}). 
\end{proof}

If we extend the theorem for $n$ numbers of composable coarse-graining, then
\begin{equation}
    S^\alpha_{\chi_1}(\rho)\ge S^\alpha_{\chi_2\chi_1} \ge \cdots \ge S^\alpha_{\chi_n\cdots\chi_2\chi_1}(\rho),
\end{equation}
where, 
\begin{equation}
    S^\alpha_{\chi_n\cdots\chi_2\chi_1}(\rho)=-\frac{1}{\alpha-1}\log{\sum_{J}p^\alpha_J V^{1-\alpha}_J}.
\end{equation}
The sequential coarse-graining is $\chi_n\cdots\chi_2\chi_1=\{\Pi_{i_1 i_2 \cdots i_n}\}$ and $J=(i_1,i_2,\cdots,i_n)$ 
and if we further extend this and connect it to the theorem (\ref{th:coarsegrained_greater}), which is discussed in Sec. (\ref{sec:cg}), then
\begin{equation}
      S^\alpha_{\chi_1}(\rho)\ge S^\alpha_{\chi_2\chi_1} \ge \cdots \ge S^\alpha_{\chi_n\cdots\chi_2\chi_1}(\rho)\ge \cdots \ge S^\alpha_R(\rho)  .
\end{equation}
After each sequential measurement, the $\alpha$-OE is decreasing and gets closer to its lower bound, which is the R\'{e}nyi entropy.

\section{Refinement of coarse-graining \label{sec:refine}}

In this Section, we would like to understand how one coarse-graining is finer compared to another, and its impact on the $\alpha$-OE. 
Refinement is making the coarse-graining finer.  It involves adding more detail to the coarse-graining. It typically involves gaining a more detailed and accurate understanding of the system. For example, in statistical mechanics, consider a gas made up of molecules with the coarse-graining involving treating the gas as a continuum and using macroscopic variables like pressure and temperature. Refinement would involve considering individual molecules, their interactions, and their velocities. When comparing different coarse-grainings, a fundamental and intuitively evident principle is that $\alpha$-OE exhibit a monotonic non-decreasing trend under stochastic post-processings of outcomes. In simple terms, if the outcome statistics of one coarse-graining $\chi$, which is a refinement of coarse-graining $\chi^\prime$, then the statistics of $\chi$ suffice for those of $\chi^\prime$. Essentially $\chi$ imparts more information than $\chi^\prime$, leading to a smaller $\alpha$-OE compared to $\chi^\prime$.

Consider two coarse-grainings $\chi$ and $\chi^\prime$, characterized by the  POVMs $\{\Pi_i\}$ and $\{\Pi_j\}$, respectively. The $\chi$ is the refinement of the coarse-graining $\chi^\prime$, denoted as $\chi \hookleftarrow \chi^\prime$, whenever there exists a stochastic matrix $m$, $m_{j|i}\ge 0$ with $ \sum_{j}m_{j|i}=1 \quad \forall \ i$, and  the corresponding set of POVMs satisfies the relation  
\begin{equation}
    \Pi_j^\prime =\sum_{i}m_{j|i}\Pi_i \quad  \forall  j .
\end{equation}

\begin{thm}\label{thm:coarse}
    For any quantum state $\rho$ and two coarse grainings $\chi$ and $\chi^\prime$, with POVM $\Pi_i$ and $\Pi_j^\prime$, respectively. If $\chi \hookleftarrow \chi^\prime$, then 
    \begin{equation}
        S^{ \alpha}_{\chi^\prime} (\rho) \geq S_\chi^\alpha (\rho),
    \end{equation}
    and equality holds if and only if,
    \begin{equation}
         \left(\frac{p_i}{V_i}\right)^\alpha=\sum_{j} m_{j|i} \left(\frac{p_j^\prime}{V_j^\prime}\right)^\alpha ,\label{eq:equality-stochastic-monotone}
    \end{equation}
    for all $i$ and  $p_j^\prime=\Tr(\Pi_j^\prime \rho)$ and $V_j^\prime=\Tr(\Pi_j^\prime)$, for $1< \alpha < \infty$.
\end{thm}

\begin{proof}
    Proof is presented in the Appendix~(\ref{app:coarse}). 
\end{proof}

The stronger version of the above theorem can be stated as follows : 

\begin{thm}
\label{thm:gen}
   For any quantum state $\rho$ and two coarse grainings $\chi$ and $\chi^\prime$, with POVM $\Pi_i$ and $\Pi_j^\prime$, respectively.  If $\chi \hookleftarrow \chi^\prime$, then 
    \begin{equation}
        S_{\chi^\prime}^\alpha (\rho) - S_\chi^\alpha (\rho) \geq \md_\alpha (P||Q),
    \end{equation}
    where $(P)^\alpha_i=p^\alpha_i$ and $(Q)^\alpha_i=\sum_j m_{j|i}\left(\frac{V_i}{V^\prime_j}p^\prime_j\right)^\alpha$, and for $1< \alpha < \infty$.
\end{thm}

\begin{proof}
    Proof is presented in the Appendix.~(\ref{app:gen}).
\end{proof}

\section{Coarse-grained state and the $\alpha$-OE \label{sec:cg}}

In this Section, we would like to study for which state $\rho$, the $\alpha$-OE will be equal to R\'{e}nyi entropy. Consider a quantum system characterized by a state $\rho$ with the coarse-graining characterized by the POVMs $\chi = \{\Pi_i\}$. The post-measurement state is given by,
\begin{equation}
    \rho^\prime=\sum_{i}\Pi_i\rho\Pi_i,
\end{equation}
 reflects the system's evolution under this measurement scenario. 
 \begin{thm} \label{thm:cg1}
 The R\'{e}nyi entropy of the post-measurement state $\rho^\prime$ is given by,
 \begin{equation}
    S^\alpha_R(\rho^\prime)=-\frac{1}{\alpha-1}\log\sum_{i}p^\alpha_i + \sum_{i}p_iS^\alpha_R(\rho_i).
\end{equation}
\end{thm}

\begin{proof}
    Proof is presented in the Appendix.~(\ref{app:cg1}). 
\end{proof}

\begin{thm} \label{thm:cg2}
    For any quantum state $\rho$ and coarse-graining $\chi=\{\Pi_i\}$, the $\alpha$-OE can be defined in terms of R\'{e}nyi entropy of the post-measurement state $\rho^\prime$ as:
    \begin{equation}
    S^\alpha_\chi(\rho)= S^\alpha_R(\rho^\prime) + \sum_{i}p_i\md_\alpha(\rho_i||\omega_i),
    \label{eq:alpha_sc_new_manu}    
    \end{equation}
where $\omega_i=\Pi_i/V_i$.
\end{thm}

\begin{proof}
    Proof is presented in the Appendix.~(\ref{app:cg2}). 
\end{proof}

For any quantum state $\rho$ and the coarse-graining $\chi$   characterized by the POVMs  {$\Pi_i$}, the associated coarse-grained state $\rho_{cg}$  is defined as~\cite{Buscemi_22}
\begin{equation}
    \rho_{cg}=\sum_{i}\frac{p_i}{V_i}\Pi_i.
\end{equation}
\begin{thm} 
    For any quantum state $\rho$ and coarse-graining $\chi=\{\Pi_i\}$
    \begin{equation}
        S_\chi^\alpha (\rho) \geq S^\alpha_R (\rho), 
    \end{equation} \label{th:coarsegrained_greater}
and equality holds if and only if $\rho = \rho_{cg} = \sum_{i}\frac{p_i}{V_i}\Pi_i$.     
\end{thm}
\begin{proof}
The positivity follows directly from the definition as it involves the relative entropy, see. The equality can be proved as follows. Let 
\begin{equation}
    S^\alpha_\chi(\rho)-S^\alpha_R(\rho)= S^\alpha_R(\rho\prime)-S^\alpha_R(\rho) + \sum_{i}p_i\md_\alpha(\rho_i||w_i),\label{eq:alpha_sc_vn-manu}
\end{equation}
 $S^\alpha_R(\rho^\prime)-S^\alpha_R(\rho)= 0$ if and only if $\rho^\prime=\rho$ which implies $\rho=\sum_{i}p_i\rho_i$, and 
 $\md_\alpha(\rho_i||\omega_i)=0$ if and only if $\rho_i=\omega_i$, hence $\rho=\sum_{i}p_iw_i = \sum_{i}\frac{p_i}{V_i}\Pi_i$.
 \end{proof}

\section{Entropy production \label{sec:entropyproduction}} 
In this Section, we study the $\alpha$-OE production in closed and open quantum systems, thereby relating it to the thermodynamic quantities. 

\subsection{Closed system} 
\begin{thm} %{$2^{nd}$ law of thermodynamics:} 
\label{thm:thmlaw}
Let  $\rho(t) = U \rho(0) U^\dagger$ be the unitarily time evolved state, and $\chi_t$ and $\chi_0$ are the coarse-graining at time $t$ and $0$. The $\alpha$-OE entropy production is defined as follows: 
\begin{equation}
        \Delta S^\alpha_{\chi_t}(\rho(t))=S^\alpha_{\chi_t}(\rho(t))-S^\alpha_{\chi_0}(\rho(0)).  
    \end{equation}
If the initial state is the coarse-grained state $\rho_{cg}$, i.e., $S^\alpha_{\chi_0}(\rho(0))=S^\alpha_R(\rho(0))$, then 
\begin{equation}
        \Delta S^\alpha_{\chi_t}(\rho(t)) \geq 0.   \label{eq:second_law}
    \end{equation}
\end{thm}
    \begin{proof}
        \begin{equation}
        \begin{split}
           & S^\alpha_{\chi_t}(\rho(t)) -S^\alpha_R(\rho(t)) \ge 0 \\
            & S^\alpha_{\chi_t}(\rho(t)) -S^\alpha_R(\rho(0)) \ge 0 \\
            & S^\alpha_{\chi_t}(\rho(t)) -S^\alpha_{\chi_0}(\rho(0)) \ge 0. \\
        \end{split}
        \end{equation}
The first inequality follows from Eq.~(\ref{eq:diff}), the second one is due to the unitary invariance of $\alpha$-R\'{e}nyi entropy, and the last line is due to the initial condition. 
    \end{proof}

Let's focus solely on the total energy as the relevant thermodynamic quantity. This system is considered homogeneous, disregarding any spatial irregularities. Let $H$ be the Hamiltonian operator corresponding to the unitary dynamics with energy eigenstates $E_k$ satisfying the eigenvalue equation:
\begin{equation}
    H\ket{E_k,g_k} = E_k \ket{E_k,g_k},
\end{equation}
where $\ket{E_k,g_k}$ are the eigenstates corresponding to the eigenvalues $E_k$ with  degree of degeneracy  $g_k$. In practical scenarios for an extremely large Hilbert space, energy measurements are not infinitely precise and come with uncertainty $\delta$. This can be mathematically represented using coarse-grained POVMs: 
\begin{equation}
    \Pi_{E} \equiv \sum_{E_k \in [E,E+\delta]} \sum_k \ketbra{E_k,g_k}.
\end{equation}

Consider a driven isolated system with a time-dependent Hamiltonian $H(\lambda_t)$, where $\lambda_t$ represents an external driving parameter such as a varying field or pressure. The corresponding coarse-grained POVMs at time $t$ are:
\begin{equation}
    \Pi_{E_t} \equiv \sum_{E_k(\lambda_t) \in [E,E+\delta]} \sum_k \ketbra{E_k(\lambda_t),g_k(\lambda_t)},
\end{equation}
where $\ket{E_k(\lambda_t), g_k(\lambda_t)}$ are the instantaneous eigenstates of $H(\lambda_t)$ within the energy shell $[E, E+\delta]$. The $\alpha$-observational entropy (OE) for such a system is defined as:
\begin{equation}
    S_{E_t}^\alpha (\rho(t)) = - \frac{1}{\alpha -1} \ln \sum_{E_t} p_{E_t}^\alpha V_{E_t}^{1-\alpha}.
    \label{eq:energy}
\end{equation}

\bla
%This can be used as the microscopic definition of thermodynamics for the isolated system where the total energy of the system is accessible. 
Here $p_{E_t}=\Tr{\Pi_{E_t}\rho(t)}$ denotes the probability of the system being in the microstate corresponding to energy $E_t$ at time $t$, and $V_{E_t}(t)=\Tr(\Pi_{E_t})$ represents the volume within the energy shell $[E, E+\delta]$ corresponding to energy $E_t$ at time $t$. The entropy production $\Delta S^\alpha_{E_t}(t)$ for energy coarse graining follows from Eq. (\ref{eq:second_law}),
\begin{equation}
    \Delta S^\alpha_{E_t}(\rho(t))=S^\alpha_{E_t}(\rho(t))-S^\alpha_{E_0}(\rho(0))\ge 0.
    \label{eq:inequality 2nd law}
\end{equation}

\subsection{Open system}
Let's consider the system interacting with the bath in a traditional open quantum system approach in which the initial state is in the product form and will be evolved by the joint unitary operator, and the state of the system is obtained by tracing out the bath's degrees. Let $H_{sb}$ denote the system-bath Hamiltonian, which can be written as,
\begin{equation}
    H_{sb} = H_s + H_b + V_{sb},
\end{equation}
where $H_s, H_b$ and $V_{sb}$ are the system Hamiltonian, bath Hamiltonian, and the system-bath interaction term.

The associated joint coarse-graining $\Pi_{s,b} = \Pi_s \otimes \Pi_b$ can be defined with $\Pi_s \equiv \ketbra{s}$ and the bath coarse-graining to be in the energy basis as $\Pi_{E_b}$ defined over the bath Hamiltonian.   $\rho_{sb}$ represents the joint state of the system and bath and  $p(s,E_b)=\Tr_{b}\{(\Pi_s\otimes E_b) \rho_{sb}\}$ represents the joint probability of the system and bath. The $\alpha$-OE corresponding to the bipartite quantum state $\rho_{sb}$ expressed as,
\begin{equation}
    S_{s_t,E_b}^\alpha(\rho_{sb}(t))=-\frac{1}{\alpha-1}\log\sum_{s_t,E_b} p^\alpha(s_t,E_b)V^{1-\alpha}(s_t,E_b)\label{eq:OE_at_time_t},
\end{equation}
where $V(s_t, E_b)=\Tr_{sb}\{\Pi_s \otimes E_b\}=V(s_t)V(E_b)$ are the volumes associated with the coarse-graining of the system and bath.

Suppose the initial state $\rho_{sb}(0) = \rho_s(0) \otimes \rho_b(0)$, and  satisfy the condition  $S^\alpha_R (\rho_{sb}(0)) = S^\alpha_{s,E_b} (\rho_{sb}(0))$, then by resorting to Theorem~(\ref{thm:thmlaw}), we can show  that the entropy production  
\begin{equation}
\begin{split}
   \xi_{1} & =  \Delta S_{s_t,E_b}^\alpha(\rho_{sb}(t))\\
   & = S^\alpha_{s_t,E_b} (\rho_{sb}(t))- S^\alpha_{s_0,E_b} (\rho_{sb}(0)) \ge 0. 
   \label{eq:dela}
   \end{split}
\end{equation}

\begin{thm} The entropy production $\xi_{2}$ satisfy the positivity condition as follows
    \begin{equation}
    \xi_{2}=\Delta S^\alpha_{s_t}(\rho_s(t)) + \Delta S^\alpha_{E_b}(\rho_B(t)) \ge 0.\label{eq:open_system_entropy}
\end{equation}
\end{thm}

\begin{proof}
    
Since the system and bath are initially uncorrelated, we have 

\begin{equation}
    S^\alpha_{s_0,E_b}(\rho_{sb}(0))=S^\alpha_{s_t}(\rho_{s}(0))+ S^\alpha_{E_b}(\rho_{B}(0)),\label{eq:initial_dec}
\end{equation}
from Eq. (\ref{eq:OE_at_time_t}) at any later time $t$ the $\alpha$-OE is given by
\begin{widetext}
 \begin{equation}
    \begin{split}
        S_{s_t,E_b}^\alpha(\rho_{sb}(t))&=-\frac{1}{\alpha-1}\log\sum_{s_t,E_b} p^\alpha(s_t,E_b)V^{1-\alpha}(s_t,E_b)\\
        &=-\frac{1}{\alpha-1}\log\sum_{s_t,E_b} p^\alpha(s_t,E_b)V^{1-\alpha}(s_t,E_b)\frac{\sum_{s_t}p^\alpha(s_t)\sum_{E_b}p^\alpha(E_b)}{\sum_{s_t}p^\alpha(s_t)\sum_{E_b}p^\alpha(E_b)}\\
        &=-\frac{1}{\alpha-1}\log\sum_{s_t}p^\alpha(s_t)V^{1-\alpha}(s_t)\frac{\sum_{s_t,E_b} p^\alpha(s_t,E_b)}{\sum_{s_t}p^\alpha(s_t)\sum_{E_b}p^\alpha(E_b)}\sum_{E_b}p^\alpha(E_b)V^{1-\alpha}(E_b)\\
        &=-\frac{1}{\alpha-1}\left(\log\sum_{s_t}p^\alpha(s_t)V^{1-\alpha}(s_t) + \log\frac{\sum_{s_t,E_b} p^\alpha(s_t,E_b)}{\sum_{s_t}p^\alpha(s_t)\sum_{E_b}p^\alpha(E_b)} + \log\sum_{E_b}p^\alpha(E_b)V^{1-\alpha}(E_b)\right)\\
        &=S^\alpha_{s_t}(\rho_{s}(t))+ S^\alpha_{E_b}(\rho_{E_b}(t)) -I^\alpha_{s_t,E_b}(\rho_{sb}(t)).\label{eq:final_dec}
    \end{split}
\end{equation}   
\end{widetext}
where $I^\alpha_{s_t,E_b}(\rho_{sb}(t))$ is the R\'{e}nyi quantum mutual information (see App.~(\ref{app:information})). From Eq.~(\ref{eq:initial_dec}) and Eq.~(\ref{eq:final_dec}), 
\begin{equation}
\begin{split}
    &\Delta S^\alpha_{s_t,E_b}(\rho_{sb}(t)) = \\
   & \Delta S^\alpha_{s_t}(\rho_s(t)) + \Delta S^\alpha_{E_b}(\rho_B(t)) - I^\alpha_{s_t,E_b}(\rho_{sb}(t)).
\end{split}
\end{equation}
The R\'{e}nyi quantum mutual information is non-negative and from Eq. (\ref{eq:dela}), we conclude
\begin{equation}
    \xi_{2}=\Delta S^\alpha_{s_t}(\rho_s(t)) + \Delta S^\alpha_{E_b}(\rho_B(t)) \ge 0 .\label{eq:open_system_entropy}
\end{equation}

\end{proof}

\subsection{Clausius' law}
Let $\Delta S_s$ denote the change in entropy of the system, and the environment is in the equilibrium temperature $T$ with $dQ$ as infinitesimal heat flow into the system, then the famous Clausius' inequality as introduced by Clausius in classical thermodynamics, is given as
\begin{equation}
    \Delta S_s - \int \frac{dQ}{T} \geq 0, 
\end{equation}
with the corresponding entropy production $\xi_3 = \Delta S_s - \int \frac{dQ}{T} \geq 0 $. In this Section, we would like to see its quantum counterpart in terms of $\alpha$-OE. 

The temperature in the case of non-equilibrium systems can be defined in various ways~\cite{casas2003temperature,muschik1977concept,muller2015thermalization,johal2009quantum,strasberg2021second}. By considering the fictitious Gibbs state $\gamma^*(\beta)$, the temperature of the state $\rho(t)$ is the temperature that the fictitious state which has the same internal energy, $\Tr (H(\lambda_t)\rho(t)) = \Tr (H(\lambda_t)\gamma^*(\beta_t))$~\cite{strasberg2021second}.  With this, the heat flux can be associated with the R\'{e}nyi entropy will be $T_t dS^\alpha_R(T_t)=dQ(t)$. Hence we have the following theorem,

\begin{thm}
    \begin{equation}
    \Delta S^\alpha_{s_t}(\rho_s(t)) -\int \frac{dQ(t)}{T_t}\ge 0. 
\end{equation}
\end{thm}

\begin{proof}
From $T_t dS^\alpha_R(T_t)=dQ(t)$, we can write
\begin{equation}
    S^\alpha_R(T_t)-S^\alpha_R(T_0)=\int \frac{dQ(t)}{T_t}.
\end{equation}
Then entropy production Eq. (\ref{eq:inequality 2nd law})
\begin{equation}
    \xi_3 (t)=S^\alpha_{E_t}(\rho(t))-S^\alpha_R(T_t) + \int \frac{dQ(t)}{T_t} + S^\alpha_R(T_0)-S^\alpha_{E_0}(\rho(0)).
\end{equation}
The isolated system prepared in the Gibbs state will have the same $\alpha$-OE and Rényi entropy, hence  
\begin{equation}
    \xi_3 (t)=S^\alpha_{E_t}(\rho(t))-S^\alpha_R(T_t) + \int \frac{dQ(t)}{T_t}.\label{eq:clausius1}
\end{equation}
The Gibbs state maximizes entropy with respect to fixed energy, and  $S^\alpha_{E_t}(\rho(t))\le S^\alpha_R(T_t)$ implies
\begin{equation}
    \int \frac{dQ(t)}{T_t}\ge \xi_3 (t) \ge 0. 
\end{equation}
Similarly, for both we have, 
\begin{equation}
    \Delta S^\alpha_{E_b}(\rho_B(t))=S^\alpha_{E_b}(\rho_B(t)) - S^\alpha_R(T_t) -\int \frac{dQ(t)}{T_t}.
\end{equation}
The convention is to count the energy flux into the system as positive; thus, for the bath, it is negative. By substituting  $\Delta S^\alpha_{E_b}(\rho_B(t))$ in Eq.~(\ref{eq:open_system_entropy})
\begin{equation}
    \begin{split}
        \xi_{2}(t)&=\Delta S^\alpha_{s_t}(\rho_S(t)) + S^\alpha_{E_b}(\rho_B(t)) - S^\alpha_R(T_t) -\int \frac{dQ(t)}{T_t}\ge 0\\
        &=\Delta S^\alpha_{s_t}(\rho_S(t)) -\int \frac{dQ(t)}{T_t}\ge + S^\alpha_R(T_t) - S^\alpha_{E_b}(\rho_B(t)),
    \end{split}
\end{equation}
and since $S^\alpha_R(T_t)\ge S^\alpha_{E_b}(\rho_B(t))$
\begin{equation}
    \xi_{3}=\Delta S^\alpha_{s_t}(\rho_S(t)) -\int \frac{dQ(t)}{T_t}\ge 0.
\end{equation}
    
\end{proof}

\subsection{Relating $\alpha$-OE to free energy}
Renyi entropy and the thermodynamic free energy have been studied in Ref~\cite{baez2022renyi}, and we would like to see their elevation to $\alpha$-OE. The Gibbs state at temperature $T$ is,
\begin{equation}
   \rho(T)=\frac{1}{Z(T)}\sum_{i}e^{-\beta E_i}\ketbra{i},
\end{equation}
where $Z(T)=\sum_{i}e^{-\beta E_i}$ is the partition function, $\beta=\frac{1}{KT}$, $K$ is the Boltzmann constant, and $E_i$ can be thought of as the energies of the physical system. The associated  Helmholtz free energy at temperature $T$ is defined as,
\begin{equation}
    A(T) = -KT\log Z(T).
\end{equation}
 
Consider the coarse graining $\chi=\{\Pi_i\}$ with constant volume term $V_i=\Tr{\Pi_i}=V$ and hence the probability $p_i=\Tr{\Pi_i \rho}=\frac{e^{E_i/T}}{Z(T)}$ corresponds to the measurement on Gibbs state. $\alpha$-OE at some temperature $T_0$ we find:
\begin{equation}
    \begin{split}
            S^\alpha_\chi (\rho_{T_0})&=-\frac{1}{\alpha-1}\log{\sum_{i}p^\alpha_i V^{1-\alpha}_i}\\
            &=-\frac{1}{\alpha-1}\log\frac{\sum_{i}e^{-\alpha E_i/T_0}}{Z^\alpha(T_0)}V^{1-\alpha}\\
            &=-\frac{1}{\alpha-1}\log\frac{\sum_{i}e^{-\alpha E_i/T_0
            }V}{V^\alpha Z^\alpha(T_0)}.
    \end{split}
\end{equation}
Here, we have taken the Boltzmann constant equal to 1.
Define the scaled partition function as $\Tilde{Z}(T)=VZ(T)$, and we have
\begin{equation}
    \begin{split}
        S^\alpha_\chi (\rho_{T_0}) &=-\frac{1}{\alpha-1}\log\frac{\sum_{i}e^{-\alpha E_i/T_0
            }V}{\Tilde{Z}^\alpha(T_0)}\\
        &=-\frac{1}{\alpha-1}\left(\log\sum_{i}e^{-\alpha E_i/T_0} V -\alpha\log \Tilde{Z}(T_0)\right).
    \end{split}
\end{equation}
Define the new  temperature $T$ as $\alpha=\frac{T_0}{T}$ and $\alpha$-OE is 
\begin{equation}
    \begin{split}
        s^\alpha_\chi(\rho_{T_0})&=-\frac{1}{\frac{T_0}{T}-1}\left(\log \Tilde{Z}(T)-\frac{T_0}{T}\log \Tilde{Z}(T_0)\right)\\
        &=-\frac{1}{T_0 - T}\left(T\log{ \Tilde{Z}(T)} - T_0 \log \Tilde{Z}(T_0)\right).
    \end{split}
\end{equation}
The rescaled Helmholtz free energy is defined as $ \Tilde{A}(T) = -T\log \Tilde{Z}(T)$, and 
\begin{equation}
    s^\alpha_\chi(\rho_{T_0}) =-\frac{\Tilde{A}(T) - \Tilde{A}(T_0)}{T - T_0}. 
\end{equation}
From quantum calculus for $\alpha \neq 1$, the Jackson derivative~\cite{akcca2019q} of a function $f(x)$ is defined as,
\begin{equation}
    \left(\frac{d}{dx}f(x)\right)_\alpha = \frac{f(\alpha x)-f(x)}{\alpha x- x},
\end{equation}
and also $\left(\frac{d}{dx}\right)_\alpha \rightarrow \frac{d}{dx}\ $ as $\alpha\rightarrow 1$. 
Thus, the $\alpha$-OE is the $\alpha^{-1}$ derivative of the negative of the rescaled Helmholtz free energy,
\begin{equation}
    S^\alpha_\chi (\rho)= -\left( \frac{d\Tilde{A}}{dT}\right)_{\alpha^{-1}}. 
\end{equation}
\bla
\section{Conclusion and Discussion}

The R\'{e}nyi entropy which is the  $\alpha$ generalization of Shannon entropy  
has been on the cornerstone of many results in both classical~\cite{renyi1965foundations,campbell1965coding,csiszar1995generalized} and quantum information theory~\cite{bennett1995generalized,hayden2003communication,cui2012entanglement, cui2012quantum,mosonyi2009generalized,audenaert2012quantum,van2002renyi}, and in quantum field theory~\cite{kudler2023renyi, kudler2020conformal, kudler2022information,kudler2021quasi,furukawa2009mutual}. Similar in spirit, we believe that our work on the generalization of observational entropy to its $\alpha$ version, $\alpha$-OE, and the proof of various properties would provide a similar path for the generalization and extensions of various results of OE.

The generalization of OE to $\alpha$-OE introduces several advantages and opens up new avenues for exploration. Traditional OE treats all probabilities in a system in an unbiased manner, assigning equal importance to states regardless of their likelihood. In contrast, $\alpha$-OE provides a parameterized framework that weights events differently based on their probabilities, offering a richer perspective. When $\alpha > 1$, $\alpha$-OE assigns higher weights to more probable events, emphasizing the dominant contributions from frequent occurrences. This makes it particularly suitable for systems where predictable or less diverse states play a significant role. Conversely, for $0 \leq \alpha \leq 1$, $\alpha$-OE highlights rare events by assigning them greater significance, making it a powerful tool for analyzing diverse systems with many low-probability states. This parameterized flexibility enhances the descriptive and analytical power of entropy, adapting to various physical and information-theoretic contexts.

We explored the properties of $\alpha$-OE, demonstrating that it retains key characteristics of OE, including its monotonic increase with the refinement of coarse-graining. It is equal to R\'{e}nyi entropy if and only if the state is a coarse-grained state and the $\alpha$-OE can only decrease with each sequential measurement. Furthermore, we established the relevance of $\alpha$-OE in various thermodynamic contexts, including entropy production and Clausius' inequality. Our findings reveal that $\alpha$-OE is not just a theoretical construct but has practical implications for quantifying entropy production in open and closed quantum systems. Additionally, we related $\alpha$-OE to Helmholtz free energy, highlighting its potential as a tool for thermodynamic analysis in quantum systems.

The introduction of $\alpha$-OE raises several interesting questions and possibilities for future research. For instance, while OE has been successfully applied to quantum chaos and many-body systems\cite{modak2022observational,sreeram2023witnessing}, exploring $\alpha$-OE in these contexts could provide deeper insights into the role of probabilistic weighting in these complex systems.  Furthermore, $\alpha$-OE represents a family of entropies parameterized by $\alpha$. A systematic study of its behavior at specific values of $\alpha$ could uncover new phenomena, such as transitions in entropy characteristics or optimal values for specific applications.

The emergence of the second law of thermodynamics is rigorously studied by recent works using OE~\cite{nagasawa2024generic,meier2024emergence}. Recently, Schindler et.al. introduced a coarse-grained entropy framework that unifies measurement-based observational entropy with Jaynes' maximum entropy principle~\cite{schindler2025unification}. It is interesting to see the extension of these results using $\alpha$-OE and the further insights that it provides.

\begin{acknowledgements}
 The Authors thank Siddhartha Das, Vinod Rao, and Ranjan Modak for fruitful discussions.  
\end{acknowledgements}

\appendix

\section{Proof of Theorem (\ref{thm:nonpos}) \label{app:nonpos}}

 \begin{proof}
        \begin{equation}
            \begin{split}
               -\dv{S^\alpha_\chi(\rho)}{\alpha}&=-\frac{1}{(\alpha-1)^2}\log{\sum_i t^\alpha_i V_i} + \frac{1}{\alpha-1}\frac{\sum_i t^\alpha_i V_i\log{t_i}}{\sum_i t^\alpha_i V_i} \\
               &=\frac{1}{(\alpha-1)^2}\left(-\log{\sum_i t^\alpha_i V_i} + \frac{\sum_i t^\alpha_i V_i\log{t^{\alpha-1}_i}}{\sum_i t^\alpha_i V_i} \right)\\
               &= \frac{1}{(\alpha-1)^2}\left(-\frac{\sum_i t^\alpha_i V_i\log{\sum_i t^\alpha_i V_i}}{\sum_i t^\alpha_i V_i} +\frac{\sum_i t^\alpha_i V_i\log{t^{\alpha-1}_i}}{\sum_i t^\alpha_i V_i} \right)\\
                &=\frac{1}{(\alpha-1)^2}\frac{\sum_i t^\alpha_i V_i \log{\frac{t^\alpha_i}{t_i\sum_i t^\alpha_i V_i}}}{\sum_i t^\alpha_i V_i}
            \end{split}
        \end{equation}
        Substituting the $x_i=\frac{t^\alpha_iV_i}{\sum_{i}t^\alpha_i V_i}$ and therefore,
        \begin{equation}
            \begin{split}
                &=\frac{1}{(\alpha-1)^2}\sum_i x_i \log{\frac{x_i}{p_i}}\\
                &=\frac{1}{(\alpha-1)^2}D(x||p).
            \end{split}
        \end{equation}

Since $D(x || p) \geq 0$, it follows that $\dv{S^\alpha_\chi(\rho)}{\alpha} \leq 0$, so $S^\alpha_\chi(\rho)$ is non-increasing with $\alpha$.
    \end{proof}

\section{Proof of Theorem.~(\ref{thm:seq}) \label{app:seq} }
\begin{proof}
Let's check the equality condition first, 
    \begin{equation}
        S^\alpha_{\chi_2\chi_1}(\rho)=-\frac{1}{\alpha-1}\log\sum_{i,j}p^\alpha_{ij}V^{1-\alpha}_{ij},
    \end{equation}
    this can be proved just by  substituting $p_{ij}=\frac{p_i}{V_i}V_{ij}$,
    \begin{equation}      
        \begin{split}
            &=\frac{1}{\alpha-1}\log\sum_{i,j}\left(\frac{V_{ij}p_i}{V_i}\right)^\alpha V^{1-\alpha}_{ij}\\
            &=-\frac{1}{\alpha-1}\log\sum_{i,j}\left(\frac{p_i}{V_i}\right)^\alpha V_{ij}\\
            &=-\frac{1}{\alpha-1}\log\sum_{i}\left(\frac{p_i}{V_i}\right)^\alpha \sum_{j}V_{ij}\\
            &=-\frac{1}{\alpha-1}\log\sum_{i}\left(\frac{p_i}{V_i}\right)^\alpha V_{i}\\
            &=S^\alpha_{\chi_1}(\rho).
        \end{split}
    \end{equation}

Let's establish the sequential coarse-graining measurement channel $\chi_2\chi_1=\{\Pi_{ij}\}$, to prove the inequality
\begin{equation}
    \e(\bullet)=\sum_{i}\Tr{\Pi_i\bullet}\ketbra{i},
\end{equation}
and,
\begin{equation}
    \e^2(\bullet)=\sum_{i,j}\Tr{\Pi_j(\Pi_i\bullet)}\ketbra{i,j},
\end{equation}
the partial trace over the second indices
\begin{equation}
    \e(\bullet)=\Tr_2{\e^2(\bullet)}.
\end{equation}
From Eq. (\ref{eq:obser_petz}), we have

    \begin{equation}
        \begin{split}
            S^\alpha_{\chi_2\chi_1}(\rho)&=-\md(\e^2(\rho)||\e^2(I))\\
            &\le -\md(\Tr_2{\e^2(\rho)}||\Tr_2{\e^2(I)})\\
            &=-\md(\e_1(\rho)||\e_1(I))\\
            &=S^\alpha_{\chi_1}(\rho).
        \end{split}
    \end{equation}
\end{proof}

\section{Proof of Theorem.~(\ref{thm:coarse})\label{app:coarse}}

\begin{proof}
For the order of $\alpha$, where $1< \alpha < \infty$,
\begin{equation} 
    \begin{split}
        S^\alpha_{\chi^\prime}(\rho)& =-\frac{1}{\alpha-1}\log{\sum_{j} p_j^{\prime\alpha} V_j^{\prime\alpha-1}}\\
        &=-\frac{1}{\alpha-1}\log{\sum_{j}\left(\frac{p_j^\prime}{V_j^\prime}\right)^\alpha V_j^\prime}\\
        &=-\frac{1}{\alpha-1}\log{\sum_{j,i}\left(\frac{p_j^\prime}{V_j^\prime}\right)^\alpha m_{j|i}V_i}.\\
    \end{split}
\end{equation}
By using Jensen's inequality, and the concavity of $\alpha$-OE for any order of $\alpha>1$,
\begin{equation}
    \begin{split}
        & \ge- \frac{1}{\alpha-1}\sum_{i}V_i\log{\sum_{j}m_{j|i}\left(\frac{p_j^\prime}{V_j^\prime}\right)^\alpha}\\
        & \ge- \frac{1}{\alpha-1}\sum_{i}V_i\log{\left(\frac{p_i}{V_i}\right)^\alpha }\label{eq:{convex}}.\\
    \end{split}
\end{equation}
By substituting $ \sum_{j}m_{j|i}\left(\frac{p_j^\prime}{V_j^\prime}\right)^\alpha=\left(\frac{p_i}{V_i}\right)^\alpha$, again applying Jensen's inequality, the above expression contains a convex function and for any value of $\alpha>1$,  the inequality sign will remain in the same direction, 
\begin{equation}
    \begin{split}
        &\ge- \frac{1}{\alpha-1}\log{\sum_{i}V_i\left(\frac{p_i}{V_i}\right)^\alpha } \label{{inequality}}\\
        & \ge S^\alpha_\chi(\rho),
    \end{split}
\end{equation}
equality holds if and only  $\sum_{j}m_{j|i}\left(\frac{p_j^\prime}{V_j^\prime}\right)^\alpha$ are equal for all $j$, and 
\begin{equation}
    \sum_{j}m_{j|i}\left(\frac{p_j^\prime}{V_j^\prime}\right)^\alpha=\left(\frac{p_i}{V_i}\right)^\alpha,
\end{equation}
and also one can see the equality of $S^{ \alpha}_{\chi^\prime} (\rho) = S_\chi^\alpha (\rho)$ at the third step of the proof, if and only if when we replace $ \sum_{j}m_{j|i}\left(\frac{p_j^\prime}{V_j^\prime}\right)^\alpha$ by $\left(\frac{p_i}{V_i}\right)^\alpha$.
\end{proof}

\section{Proof of Theorem.~(\ref{thm:gen}). \label{app:gen}}

\begin{proof}
\begin{equation}
    \begin{split}
        S^\alpha_{\chi^\prime}(\rho) - S^\alpha_{\chi}(\rho)&=\frac{1}{\alpha-1}\log{\left(\frac{\sum_{i}\left(\frac{p_i}{V_i}\right)^\alpha V_i}{\sum_{j}\left(\frac{p_j^\prime}{V_j^\prime}\right)^\alpha V^\prime_j}\right)}\\
        &=\frac{1}{\alpha-1}\log{\left(\frac{\sum_{i,j}p^\prime_j\left(\frac{p_i}{V_i}\right)^\alpha V_i}{\sum_{j,i}p_i\left(\frac{p_j^\prime}{V_j^\prime}\right)^\alpha V^\prime_j}\right)}
    \end{split}
\end{equation}

\begin{equation}
            =\frac{1}{\alpha-1}\frac{\sum_{i,j}p^\prime_j\left(\frac{p_i}{V_i}\right)^\alpha V_i}{\sum_{i,j}p^\prime_j\left(\frac{p_i}{V_i}\right)^\alpha V_i}\log{\left(\frac{\sum_{i,j}p^\prime_j\left(\frac{p_i}{V_i}\right)^\alpha V_i}{\sum_{j,i}p_i\left(\frac{p_j^\prime}{V_j^\prime}\right)^\alpha V^\prime_j}\right)}.
\end{equation}
 From the concavity property of $\alpha$-OE, and apply the log-sum inequality.
\begin{equation}
        \ge \frac{1}{\alpha-1}\frac{\sum_{i,j}p^\prime_j\left(\frac{p_i}{V_i}\right)^\alpha V_i}{\sum_{i,j}p^\prime_j\left(\frac{p_i}{V_i}\right)^\alpha V_i}\log{\left(\frac{p^\prime_j\left(\frac{p_i}{V_i}\right)^\alpha V_i}{p_i\left(\frac{p_j^\prime}{V_j^\prime}\right)^\alpha V^\prime_j}\right)}\\
\end{equation}

\begin{equation}
    \ge \frac{1}{\alpha-1}\frac{\sum_{i,j}p^\prime_j\left(\frac{p_i}{V_i}\right)^\alpha V_i}{\sum_{i,j}p^\prime_j\left(\frac{p_i}{V_i}\right)^\alpha V_i}\log{\left(\frac{m_{j|i}p^\alpha_i \left(\frac{p^\prime_jV_i}{V^\prime_j}\right)^{1-\alpha}}{m_{j|i}p_i}\right)}.\\
\end{equation}
Apply the log-sum inequality again and take the summation inside; the inequality sign will not change and will remain in the same direction, then

\begin{equation}
    \begin{split}
        &\ge \frac{1}{\alpha-1}\frac{\sum_{i,j}p^\prime_j\left(\frac{p_i}{V_i}\right)^\alpha V_i}{\sum_{i,j}p^\prime_j\left(\frac{p_i}{V_i}\right)^\alpha V_i}\log{\left(\frac{\sum_{i}p^\alpha_i \sum_{j}m_{j|i}\left(\frac{p^\prime_jV_i}{V^\prime_j}\right)^{1-\alpha}}{\sum_{i,j}m_{j|i}p_i}\right)}\\
         &\ge \frac{1}{\alpha-1}\log{\sum_{i}P^\alpha_iQ^{1-\alpha}_i}\\
         &\ge \md_\alpha(P||Q)\nonumber.
    \end{split}
\end{equation}
\end{proof}

\section{Proof of Theorem.~(\ref{thm:cg1}).\label{app:cg1}}

\begin{proof}

Let $\rho^\prime=\sum_{i}p_i\rho_i$ with $\rho_i=\frac{\Pi_i\rho\Pi_i}{p_i}$. The orthogonal decomposition  $\rho_i=\sum_{j}q^i_j |e^i_j\rangle \langle e^i_j|$ and  hence $\rho^\prime = \sum_{i,j}p_iq^i_j|e^i_j\rangle \langle e^i_j|$. Consider
    \begin{equation}
        \begin{split}
            &S^\alpha_R(\rho^\prime)=-\frac{1}{\alpha-1}\log\sum_{i}p^\alpha_i + \sum_{i}p_iS^\alpha_R(\rho_i)\\
            &=-\frac{\sum_{i}p_i}{\alpha-1}\log{\sum_{i}p^\alpha_i} -\sum_{i}p_i\frac{1}{\alpha-1}\log{\sum_{j}q^{i\alpha}_j}
        \end{split}
    \end{equation}
    \begin{equation}
        \begin{split}
            &=-\frac{\sum_{i}p_i}{\alpha-1}\left(\log\sum_{i}p^\alpha_i + \log\sum_{j}q^{i\alpha}_j\right)\\
            &=-\frac{\sum_{i}p_i}{\alpha-1}\log\left(\sum_{i}p^\alpha_i\sum_{j}q^{i\alpha}_j\right)\\
            &=-\frac{1}{\alpha-1}\log{\sum_{i,j}p^\alpha_i q^{i\alpha}_j} \\
            &= S^\alpha_R(\rho^\prime).
        \end{split}
    \end{equation}
\end{proof}

\section{Proof of Theorem.~(\ref{thm:cg2} \label{app:cg2}}

%\begin{proof}
\begin{widetext}
\begin{equation}
    \begin{split}
        S^\alpha_R(\rho^\prime) + \sum_{i}p_i\md_\alpha(\rho_i||\omega_i)
        &=S^\alpha_R(\rho^\prime)+\sum_{i}p_i\frac{1}{\alpha-1}\log{\sum_{j}q^{i\alpha}_j\left(\frac{1}{V_i}\right)^{1-\alpha}}\\
        &=-\frac{\sum_{i}p_i}{\alpha-1}\log\sum_{i}p^\alpha_i + \sum_{i}p_iS^\alpha_R(\rho_i)+\frac{\sum_{i}p_i}{\alpha-1}\left(\log\sum_{j}q^{i\alpha}_j + \log{\left(\frac{1}{V_i}\right)^{1-\alpha}}\right)\\
        &=-\frac{\sum_{i}p_i}{\alpha-1}\log\sum_{i}p^\alpha_i + \sum_{i}p_iS^\alpha_R(\rho_i)-\sum_{i}p_iS^\alpha_R(\rho_i)+\frac{\sum_{i}p_i}{\alpha-1}\log{\left(\frac{1}{V_i}\right)^{1-\alpha}}\\
        &=-\frac{\sum_{i}p_i}{\alpha-1}\log\sum_{i}p^\alpha_i-\frac{\sum_{i}p_i}{\alpha-1}\log{V^{1-\alpha}_i}\\
        &=-\frac{\sum_{i}p_i}{\alpha-1}\log\sum_{i}p^\alpha_i V^{1-\alpha}\\
        &=-\frac{1}{\alpha-1}\log{\sum_{i}p^\alpha_iV^{1-\alpha}_i}.
    \end{split}
\end{equation}
\end{widetext}
%\end{proof}

\section{R\'{e}nyi quantum mutual information}
\label{app:information}
Quantum mutual information is a measure of the correlations between two subsystems of a quantum system. It quantifies the amount of information that two subsystems share about each other.
Given a bipartite quantum state $\rho_{AB}$ representing the joint state of two subsystems $A$ and $B$, the quantum mutual information $I(A;B)$ is defined as:
\begin{equation}
    I^\alpha_{A;B}(\rho_{AB})=S^\alpha_R(\rho_A) + S^\alpha_R(\rho_B) - S^\alpha_R(\rho_{AB})
\end{equation}
One important property of R\'{e}nyi quantum mutual information $I^\alpha_{A;B}$ is its non-negativity. In quantum mechanics, mutual information is always greater than or equal to zero, which implies that the information shared between two subsystems $A$ and $B$ cannot be negative. This property reflects the fact that correlations between quantum systems cannot reduce the total amount of information shared between them.

Furthermore, the R\'{e}nyi Quantum mutual information can also be written in terms of Petz R\'{e}nyi relative entropy \cite{berta2015renyi} as
\begin{equation}
    I^\alpha_{A;B}(\rho_{AB})=\md_\alpha(\rho_{AB}||\rho_A \otimes \rho_B)
\end{equation}
Intuitively, R\'{e}nyi quantum mutual information measures how much knowing the value of one random variable reduces uncertainty about the other random variable. If $A$ and $B$ are independent, then $I^\alpha_{A;B}=0$, indicating that knowing $A$ provides no information about $B$ and vice versa. Conversely, if $A$ and $B$ are perfectly dependent, then $I^\alpha_{A; B}$ is maximized, indicating that knowing the value of one variable fully determines the value of the other.

In our case when we are taking the bipartite quantum state $\rho_{sb}$ of the system and bath, the R\'{e}nyi quantum mutual information of the system and bath follows that

\begin{equation}
    \begin{split}
        I^\alpha_{s_t,E_b}(\rho_{sb}(t))&=S^\alpha_R(\rho_S (t)) + S^\alpha_R (\rho_B (t)) - S^\alpha_R(\rho_{sb}(t))\\
        &=\frac{1}{\alpha-1}\left(-\log\sum_{s_t}p^\alpha_{s_t} -\log\sum_{E_b}p^\alpha_{E_b} + \log\sum_{s_t,E_b}p^\alpha_{s_t,E_b}\right)\\
        &=\frac{1}{\alpha-1}\log \frac{\sum_{s_t,E_b}p^\alpha_{s_t,E_b}}{\sum_{s_t}p^\alpha_{s_t}\sum_{E_b}p^\alpha_{E_b}}
    \end{split}
\end{equation}

Where $S^\alpha(\rho_S (t))$ and $S^\alpha (\rho_B (t))$ represent the Rényi entropy of the system and bath at time $t$, respectively, and $ S^\alpha(\rho_{sb}(t))$ denotes the Rényi entropy of the combined bipartite state of the system and bath at time $t$.

\hspace{1cm}
%\bibliography{chaos,thermo,rel_ent,renyi}

\end{document}